\documentclass[11pt]{article}

\usepackage{thumbpdf, amssymb, amsmath, amsthm, microtype, array,
graphicx, verbatim, listings, color, fancybox, enumerate}
\usepackage{hyperref}
\usepackage{mathtools}
\usepackage{booktabs}
\usepackage{setspace}

\usepackage{boxedminipage}
\usepackage[linesnumbered,ruled]{algorithm2e}

\usepackage{caption}
\usepackage{subcaption}
\usepackage{thmtools}
\usepackage{bbm}
\usepackage{pgf,tikz}
\usetikzlibrary{arrows}
\usepackage[margin=1.2in]{geometry}
\usepackage{multicol,multirow}

\usepackage{mathrsfs}

\usepackage{thm-restate}
\usepackage{cleveref}

\usepackage{amsfonts,amstext}

\makeatletter
\newcommand{\leqnomode}{\tagsleft@true\let\veqno\@@leqno}
\newcommand{\reqnomode}{\tagsleft@false\let\veqno\@@eqno}
\makeatother

\mathchardef\mhyphen="2D

\newcommand{\R}{\ensuremath{\mathbb R}}
\newcommand{\Rp}{\ensuremath{\mathbb R_{\geq 0}}}

\newcommand{\Zp}{\ensuremath{\mathbb Z_{\geq 0}}}

\newcommand{\twoecm}{\ensuremath{\mathsf{2ECM}}}
\newcommand{\twoecmip}{\ensuremath{\twoecm\mhyphen\mathsf{IP}}}
\newcommand{\twoecmlp}{\ensuremath{\twoecm\mhyphen\mathsf{LP}}}
\newcommand{\tsplp}{\ensuremath{\mathsf{Subtour\mhyphen LP}}}
\newcommand{\tsp}{\ensuremath{\mathsf{TSP}}}
\newcommand{\twoecs}{\ensuremath{\mathsf{2ECS}}}
\newcommand{\twoecmig}{\ensuremath{\alpha\mathsf{2ECM}}}

\newcommand{\adm}{\ensuremath{\mathsf{A_{e}}}}

\newcommand{\hG}{\hat{G}}

\newcommand{\hH}{\hat{H}}
\newcommand{\bG}{\overline{G}}
\newcommand{\bV}{\overline{V}}
\newcommand{\bE}{\overline{E}}
\newcommand{\bH}{\overline{H}}
\newcommand{\cc}{\mu}
\newcommand{\ec}{\lambda}
\newcommand{\ccvec}{q}

\newtheorem{theorem}{Theorem}
\newtheorem{lemma}[theorem]{Lemma}

\newtheorem{definition}[theorem]{Definition}

{\theoremstyle{remark} }

\newenvironment{proofof}[1]{\begin{proof}[Proof of #1]}{\end{proof}}

\begin{document}

\title{A $4/3$-Approximation Algorithm for the Minimum $2$-Edge Connected
Multisubgraph Problem in the Half-Integral Case\footnote{A preliminary version of this paper will appear in the Proceedings of APPROX 2020.} }

\author{
    Sylvia Boyd \thanks{{\tt sboyd@uottawa.ca}.
    School of Electrical Engineering and Computer Science, University of Ottawa, Ottawa, Canada} 
    \and
    Joseph Cheriyan \thanks{{\tt \{jcheriyan,sharat.ibrahimpur\}@uwaterloo.ca}.
    Dept. of Combinatorics and Optimization, University of Waterloo, Waterloo, Canada} 
    \and \addtocounter{footnote}{-1} 
    Robert Cummings \footnotemark 
    \and \addtocounter{footnote}{-1}
    Logan Grout \footnotemark
    \and \addtocounter{footnote}{-1}
    Sharat Ibrahimpur \footnotemark 
    \and
    Zolt\'an Szigeti \thanks{{\tt zoltan.szigeti@grenoble-inp.fr}. University Grenoble Alpes, CNRS, G-SCOP, Grenoble, France}
    \and \addtocounter{footnote}{-2}
    Lu Wang \footnotemark
}

\maketitle

\begin{abstract}

Given a connected undirected graph $\bG$ on $n$ vertices, and non-negative edge costs $c$, the $\twoecm$ problem is that of finding a $2$-edge~connected spanning multisubgraph of $\bG$ of minimum cost. 
The natural linear program (LP) for $\twoecm$, which coincides with the subtour LP for the Traveling Salesman Problem on the metric closure of $\bG$, gives a lower bound on the optimal cost.
For instances where this LP is optimized by a half-integral solution $x$, Carr and Ravi (1998) showed that the integrality gap is at most $\frac43$: they show that the vector $\frac43 x$ dominates a convex combination of incidence vectors of $2$-edge connected spanning multisubgraphs of $\bG$.

We present a simpler proof of the result due to Carr and Ravi by
applying an extension of Lov\'{a}sz's splitting-off theorem.
Our proof naturally leads to a $\frac43$-approximation algorithm for half-integral instances. 
Given a half-integral solution $x$ to the LP for $\twoecm$, we give
an $O(n^2)$-time algorithm to obtain a $2$-edge connected spanning
multisubgraph of $\bG$ whose cost is at most $\frac43 c^T x$.

\end{abstract}

\section{Introduction} \label{intro}

The $2$-edge connected multisubgraph ($\twoecm$) problem is a
fundamental problem in survivable network design where one wants
to be resilient against a single edge failure.  In this problem,
we are given an undirected graph $\bG = (\bV,\bE)$ with non-negative edge costs
$c$ and we want to find a $2$-edge connected spanning multisubgraph
of $\bG$ of minimum cost. Below we give an integer linear program
for $\twoecm$. The variable $x_e$ denotes the number of copies of edge
$e$ that are used in a feasible solution. 
For any $S \subset V$, $\delta(S) := \{ e = uv \in \bE: u \in S, v \notin S \}$ denotes the cut induced by $S$.
For any $F \subseteq E$ and vector $x\in \R^E$, we use $x(F)$ as a shorthand for $\sum_{e \in F} x_e$.  Also, for any graph $H$ with edge costs $c$, we sometimes use $c(H)$ as a shorthand for $c(E(H))$.  

\begin{minipage}{0.15\textwidth}
\leqnomode
\begin{equation}
\tag{\twoecmip} \label{eq:2ecip}
\end{equation}
\reqnomode
\end{minipage}
\
\begin{minipage}{0.8\textwidth}
\begin{alignat}{3}
\min \qquad && \sum_{e \in \bE} c_e x_e \qquad && \label{eq:obj} \\
\text{subject to} \qquad && x(\delta(S)) \geq 2 \qquad && \forall \, \emptyset \subsetneq S \subsetneq \bV, \label{eq:subtour} \\
&& x_e \geq 0 \qquad && \forall e \in \bE, \label{eq:nonneg} \\
&& x_e \text{ integral} \qquad && \forall e \in \bE . \label{eq:2ecint}
\end{alignat}
\end{minipage}

\medskip

It is easy to see that an optimal solution for $\twoecm$ never has
to use more than two copies of an edge.
As is discussed in \cite{CarrR98}, since we are allowed to use more than one copy of an edge, without loss of generality, we may assume that $\bG$ is complete by performing the metric completion: for each $u,v \in \bV$ we set the new cost of the edge $uv$ to be the shortest path distance between $u$ and $v$ in $\bG$.
In the sequel, we assume that $\bG$ is a complete graph and that the cost function $c$ is metric i.e., $c \geq 0$ and for every $u,v,w \in \bV$, we have $c_{uw} \leq c_{uv} + c_{vw}$.

The linear relaxation $(\twoecmlp)$ for $\twoecm$ is obtained by dropping the integrality constraints given by \eqref{eq:2ecint}.
By a result due to Goemans and Bertsimas \cite{GoemansB93} called the \emph{parsimonious property}, adding the constraint $x(\delta(v)) = 2$ for each $v \in \bV$ to $(\twoecmlp)$ does not increase the optimal solution value; here, we require the assumption that the costs form a metric.
So, the optimal value of $(\twoecmlp)$ is the same as the optimal value for the well-known subtour elimination LP $(\tsplp)$ for the Traveling Salesman Problem ($\tsp$) defined below.
Due to this connection, we often refer to an optimal solution for $(\twoecmlp)$ as an optimal solution to (${\tsplp}$), and vice versa.
Another consequence of the parsimonious property is that for graphs with at least $3$ vertices, the constraint $x_e \leq 1$ is implied by ($\tsplp$): for any $e = uv$, we have ${2 x_e = x(\delta(u))+x(\delta(v))-x(\delta(\{u,v\})) \leq 2}$.

\begin{minipage}{0.15\textwidth}
\leqnomode
\begin{equation}
\tag{\tsplp} \label{eq:tsplp}
\end{equation}
\reqnomode
\end{minipage}
\
\begin{minipage}{0.8\textwidth}
\begin{alignat}{3}
\min \qquad && \sum_{e \in \bE} c_e x_e \qquad && \label{eq:objtsp} \\
\text{subject to} \qquad && x(\delta(S)) \geq 2 \qquad && \forall \, \emptyset \subsetneq S \subsetneq \bV, \label{eq:subtourtsp} \\
\qquad && x(\delta(v)) = 2 \qquad && \forall v \in \bV, \label{eq:degtsp} \\
&& x_e \geq 0 \qquad && \forall e \in \bE. \label{eq:nonnegtsp}
\end{alignat}
\end{minipage}

\medskip

A long-standing open problem called the ``four-thirds conjecture'' states that the integrality gap of $(\tsplp)$ is $\frac43$.
Besides the importance of $\twoecm$ in the field of survivable network design, the connection between $(\twoecmlp)$ and $(\tsplp)$ has spurred interest in determining the integrality gap for $(\twoecmlp)$ as a means to gaining useful lower bounds on the integrality gap for $(\tsplp)$.
The general version of metric $\tsp$ has resisted all attempts at proving an upper bound better than $\frac32$ on the integrality gap, so a great deal of research has focused on obtaining improvements for important special cases.
In \cite{SchalekampWZ14}, the authors conjecture that the integrality gap for $(\tsplp)$ is achieved on instances where an optimal (fractional) solution to $(\tsplp)$
is half integral i.e., $2 x_e \in \Zp$ for all $e$.
We refer to such instances as \emph{half integral instances}.
More than two decades ago, Carr and Ravi \cite{CarrR98} proved that the integrality gap of $(\twoecmlp)$ is at most $\frac43$ in the half-integral case.
They show that $\frac43 x$ dominates a convex combination of $2$-edge connected spanning multisubgraphs of $\bG$.
This supports the four-thirds conjecture for $\tsp$ since the (integer) optimal value for $\twoecm$ lower bounds the (integer) optimal value for $\tsp$.
However, the proof of Carr and Ravi does not give a polynomial-time algorithm for $\twoecm$.
Very recently, in \cite{KarlinKG20}, Karlin, Klein, and Oveis Gharan gave a randomized approximation algorithm for half-integral instances of $\tsp$ whose (expected) approximation factor is $\frac32-0.00007$.
This immediately implies a better than $\frac32$-approximation algorithm,
albeit randomized, for $\twoecm$ as well.

We note that the result of Carr and Ravi mentioned above does not apply to
the strict variant of $\twoecm$  (henceforth denoted by $\twoecs$) where we are allowed
to pick at most one copy of an edge in $\bG$, i.e. where we are considering subgraphs of $\bG$ rather than multisubgraphs; similarly, our main result
does not apply to $\twoecs$.

\subsection{Our Work} \label{ourwork}

Our main contribution is a deterministic approximation algorithm
for $\twoecm$ on half-integral instances that matches the existence
result in \cite{CarrR98}.

\begin{theorem} \label{thm:fourthirdsalgo}
Let $x$ denote an optimal half-integral solution to an instance $(\bG,c)$ of $(\tsplp)$ (and $(\twoecmlp)$). 
There is an $O(|V(\bG)|^2)$-time algorithm for computing a $2$-edge connected spanning multisubgraph of $\bG$ with cost at most $\frac43 c^T x$.
\end{theorem}

We can strengthen the above result by using the meta-rounding algorithm of Carr and Vempala \cite{CarrV02}. 
Under some mild assumptions, the meta-algorithm uses an LP-based $\alpha$-approximation algorithm as a black-box and gives an efficient procedure to obtain a convex combination of integer solutions that is dominated by $\alpha x$, where $x$ is a (feasible) fractional LP solution.
We defer our proof of the following result to a subsequent full version of our paper.
 
\begin{quote}
Let $\bG = (\bV,\bE)$ be a complete graph on $n$ vertices.
Let $x \in \Rp^{\bE}$ be a fractional half-integral solution to $(\tsplp)$ (or equivalently, $(\twoecmlp)$) i.e., $x$ satisfies the constraints \eqref{eq:subtourtsp}-\eqref{eq:nonnegtsp}.
In $\mathsf{poly}(n)$-time, we can obtain $2$-edge connected spanning multisubgraphs $H_1,\dots,H_k$ and nonnegative real numbers $\cc_1,\dots,\cc_k$, $\sum_{i=1}^k \cc_i = 1$ satisfying $\sum_{i=1}^k \cc_i \chi^{E(H_i)} \leq \frac43 x$.
\end{quote}

Given a half-integral solution $x$ to $(\tsplp)$ for $\bG$, let $G = (V,E)$ denote the multigraph induced by $2x$.
Formally, the vertex-set $V := \bV$, and for each edge $e \in \bE$, the edge-set $E$ has $2 x_e$ copies of the edge $e$.
Note that if $|\bV| \geq 3$, then $2 x_e \in \{0,1,2\}$ for all $e \in \bE$, 
and if $|\bV| = 2$, then $2 x_e = 4$ for the unique edge $e \in \bE$. 
With a slight abuse of notation, we use the same cost function $c$
to denote the edge costs in $G$ i.e., $c_f := c_e$ where $e \in \bE$
gave rise to the edge $f \in E$.
By \eqref{eq:degtsp} and \eqref{eq:subtourtsp}, $G$ is a $4$-regular $4$-edge connected multigraph.
Theorem~\ref{thm:fourthirdsalgo} follows from the following result applied to the graph $G$ induced by $2x$.

\begin{theorem} \label{thm:twothirds}
Let $G = (V,E)$ be a $4$-regular $4$-edge connected multigraph on $n$ vertices. 
Let $c : E \to \R$ be an arbitrary cost function on the edges of $G$ (negative costs on the edges are allowed), and let $e$ be an arbitrary edge in $G$.
Then, in $O(n^2)$ time, we can find a $2$-edge connected spanning subgraph
$H$ of $G - e$ satisfying:
\begin{enumerate}[(i)]
\item $c(H) \leq \frac23 c(G-e)$; and
\item each multiedge of $G$ appears at most once in $H$ (multiedges may arise in $H$ due to multiedges in $G$).
\end{enumerate}
\end{theorem}

For any $F \subseteq E$, let $\chi^F \in \{0,1\}^E$ denote the characteristic
vector of $F$: $\chi^F_e = 1$ if and only if $e \in F$.  Note that
distinct multiedges in $E$ correspond to distinct coordinates in
$\chi^F$.  As mentioned before, Carr and Ravi prove the existence
of such a subgraph $H$ by showing that for any $4$-regular $4$-edge
connected multigraph $G$, there exists a finite collection
$H_1,\dots,H_k$ of $2$-edge connected spanning subgraphs of $G$ such that
$\frac23 \chi^{E(G) \setminus \{e\}}$ lies in the convex hull of
$\{\chi^{H_i}\}_i$. At a high level, their proof is inductive and
splits into two cases based on whether $G$ has a certain kind of a
tight set (a cut of size $4$).  In the first case they construct
two smaller instances of the problem by contracting each of the
shores of the tight set, and in the second case they perform two
distinct  splitting-off operations at a designated vertex to obtain
two smaller instances of the problem.  In either case, the convex
combinations from the two subinstances are merged to obtain a
convex combination for $G$. 
The first case requires gluing since the subgraphs obtained from the two subinstances need to agree on a (tight) cut. 
Merging the convex combinations arising from the second case is rather straightforward as the two subinstances are more or less independent.

Our first insight in this work is that the case from Carr and
Ravi's proof that requires the gluing step can be completely avoided, thereby unifying the analysis. 
This is discussed in Section \ref{covers}.
Our proof relies on an extension of Lov\'{a}sz's splitting-off
theorem that is due to Bang-Jensen et~al., \cite{BangJensenGJS99}.
For further discussion on splitting-off theorems, see
\cite[Chapter~8]{frank:book}.
The challenge in efficiently finding a cheap subgraph $H$ from the above convex combination construction is that each inductive step requires solving two subinstances of the problem, each with one fewer vertex, leading to an exponential-time algorithm.
Having said that, an (expected) polynomial-time Las Vegas randomized algorithm can be easily designed that randomly recurses on one of the two subinstances and produces a $2$-edge connected spanning subgraph whose expected cost is at most $\frac23 c(G-e)$.
Our second insight, which is used in derandomizing the above procedure, is that it is easy to recognize which of the two subinstances leads to a ``cheaper'' solution, so we recurse only on the cheaper subinstance.
Complementing this step, we lift the solution back to the original instance.
This operation can lead to two different outcomes so the cost analysis must account for the worst outcome.
There is a choice of defining the costs in the subinstance such that the cost of the lifted subgraph is the same irrespective of the outcome.
Such a choice can lead to negative costs, but this is not a hindrance for our inductive step because Theorem~\ref{thm:twothirds} allows arbitrary real-valued edge costs. 
This generality of cost functions is crucial to our algorithm.

In Section~\ref{sec:7over8} we consider a well-studied special case of the $\twoecm$ problem. 
We present a simple $O(n^3)$-time algorithm that given a $3$-regular $3$-edge
connected graph $G$, finds a $2$-edge connected spanning multisubgraph of cost at most $\frac78 c(G)$ (see Theorem~\ref{thm:7over8}).
The proof is inspired by that of Haddadan, Newman, and Ravi in \cite{HaddadanNR19} where they give a polynomial-time algorithm for this problem with a factor $\frac{15}{17} (> \frac78)$.
In \cite[Theorem~1.1]{HaddadanN18}, Haddadan and Newman improve this result to a factor $\frac78$, and very recently, in \cite[Theorem~1.20]{Haddadan:thesis}, Haddadan claims a stronger result with a factor of $\frac{41}{47} = \frac78 - \frac{1}{376}$.
We remark that these proofs are longer and/or more complicated than that of Theorem~\ref{thm:7over8}.
Another motivation for Section~\ref{sec:7over8} is to illustrate the potential of Theorem~\ref{thm:twothirds} in giving simpler proofs for results that may not have any explicit half-integrality restrictions.

\subsection{Related Work} \label{related}

The $\twoecm$ problem has been intensively studied in network design and several works have tried to bound the integrality gap $\twoecmig$ of $(\twoecmlp)$.
For the general case with metric costs, we have $\frac65 \leq \twoecmig \leq \frac32$, where the lower bound is from \cite{AlexanderBE06} and the upper bound follows from the polyhedral analysis of Wolsey \cite{Wolsey80} and Shmoys and Williamson \cite{ShmoysW90} (this analysis also gives a $\frac32$-approximation algorithm).  It is generally conjectured that $\twoecmig = \frac43$, however in \cite{AlexanderBE06}, Alexander et~al., study $\twoecmig$ and conjecture that $\twoecmig = \frac65$ based on their findings.   
As mentioned before, Carr and Ravi \cite{CarrR98} show that the integrality gap of $(\twoecmlp)$ is at most $\frac43$ in the half-integral case.  In 
\cite{BoydL17} Boyd and Legault consider a more restrictive collection of instances called half-triangle instances where the optimal LP solution is half-integral and the graph induced by the half-edges is a collection of disjoint triangles.
They prove that $\twoecmig = \frac65$ in this setting.
Half-triangle solutions are of interest as there is evidence that the integrality gap of $(\twoecmlp)$ is attained at such solutions (see \cite{AlexanderBE06}).
When the costs come from a graphic metric (i.e., we want to find a minimum-size $2$-edge connected spanning multisubgraph of a given unweighted graph), we have $\frac87 \leq \twoecmig \leq \frac43$ (see \cite{BoydFS16,SeboV14}).

\section{A Simpler Proof of a Result of Carr and Ravi} \label{covers}

In this section, we give a simplified proof of the following result
from \cite{CarrR98}.
As mentioned before, avoiding the case involving the gluing operation is useful for our algorithm in Section~\ref{algo}.
For notational convenience, for any subgraph $K$ of some graph, we use $\chi^K$ to denote $\chi^{E(K)}$ whenever the underlying graph is clear from the context. 

\begin{theorem}[Statement~1 from \cite{CarrR98}] \label{uniformcover}
Let $G = (V,E)$ be a $4$-regular $4$-edge connected multigraph and $e = uv$ be an arbitrary edge in this graph.
There exists a finite collection $\{H_1,\dots,H_k\}$ of $2$-edge connected spanning subgraphs of $G-e$ such that for some nonnegative $\cc_1,\dots,\cc_k$ with $\sum_i \cc_i = 1$, we have $\frac23 \chi^{E \setminus \{e\}} = \sum_{i=1}^k \cc_i \chi^{H_i}$.
Moreover, we may assume that none of the $H_i$'s use more than one copy of an edge in $E$; $H_i$ may have multiedges as long as they come from distinct edges in $G$.
\end{theorem}

\subsection{Operations involving splitting-off at a vertex} \label{sec:splitoff}

The following tools on the splitting-off operation will be useful.
In keeping with standard terminology, we designate a vertex $v$ (one of the endpoints of $e$ in the theorem statement) at which the splitting-off operation is applied.
For a multigraph $H=(V,E)$ and $x,y \in V$, let $\ec_H(x,y)$ denote the size of a minimum $(x,y)$-cut in $H$, and let $\deg_H(x)$ denote the degree of $x$ in $H$. Note that each multiedge is counted separately towards the degree of a vertex and the size of a cut.

\begin{definition} \label{splitoffdefn}
Given a multigraph $G$ and two edges $sv$ and $vt$ that share an endpoint $v$, the graph $G_{s,t}$ obtained by \emph{splitting off} the pair $(sv,vt)$ at $v$ is given by $G + st - sv - vt$.
\end{definition}

\begin{definition}
	Given a multigraph $G$ and a vertex $v$ of $G$ of even degree, a \emph {complete splitting at $v$} is a sequence of $\frac12 deg_G(v)$ splitting off operations that result in vertex $v$ having degree zero in the resulting graph.
\end{definition}

\begin{definition} \label{admissibledefn}
Let $k \geq 2$ be an integer and let $G$ be a multigraph such that for all $x,y \in V \setminus \{v\}$, ${\ec_G(x,y) \geq k}$.
Let $e =sv$ and $vt$ be two edges incident to $v$.
We say that the pair $(sv,vt)$ is \emph{admissible} if for all $x,y \in V \setminus \{v\}$, $\ec_{G_{s,t}}(x,y) \geq k$, and for a particular edge $e \in \delta(v)$, we let $A_e$ denote the set of edges $f\in \delta(v)\setminus \{e\}$ such that $(e,f)$ is an admissible pair.
\end{definition}

The following result due to Bang-Jensen et~al., \cite{BangJensenGJS99} shows that in our setting with a $4$-regular $4$-edge connected multigraph at least two distinct edges incident to $v$ form an admissible pair with $e = uv$. 
Using this we can perform a complete splitting at $v$ in two distinct ways. 

\begin{lemma}[Theorem~2.12 from \cite{BangJensenGJS99}] \label{bangjensenlemma}
Let $k \geq 2$ be an even integer. 
Let $G$ be a multigraph such that for all $x,y \in V \setminus \{v\}$, $\ec_{G}(x,y) \geq k$.
Let $\deg_G(v)$ be even (each multiedge is counted separately towards the degree).
Then, $|A_{uv}| \geq \frac12 \deg_G(v)$.
\end{lemma}

\begin{lemma} \label{lem:twooptions}
Let $G$ be a $4$-regular $4$-edge connected multigraph and $e=vx$ be an edge incident to $v$.
Then, (i) $|A_e| \ge 2$; and (ii) if $(e,f)$ is an admissible pair for some $f = vy \in \delta(v)\setminus \{e\}$, then the remaining two edges in $\delta(v)\setminus \{e,f\}$ form an admissible pair in $G_{x,y}$.
\end{lemma}

\begin{proof}
Conclusion (i) follows from Lemma~\ref{bangjensenlemma} since $G$ is $4$-regular and $4$-edge~connected.
For conclusion (ii), let $f \in \delta(v)\setminus \{e\}$ be such that $(e,f)$ forms an admissible pair in $G$.
Let $G_{x,y}$ denote the graph obtained by splitting off the pair
$(e=vx,f=vy)$ i.e., $G_{x,y} = G - vx - vy + xy$.
Observe that the hypothesis of Lemma~\ref{bangjensenlemma} still
holds for $G_{x,y}$ with $k=4$ because (a) we performed a splitting off operation 
using an admissible pair of edges; and (b) $\deg_{G_{x,y}}(v) = 2$ is even.
Let $g$ denote one of the two remaining edges in $\delta(v)\setminus \{e,f\}$. 
By Lemma~\ref{bangjensenlemma}, the other unique edge $h \in \delta(v)\setminus \{e,f,g\}$ forms an admissible pair with $g$ in $G_{x,y}$.
\end{proof}

Equipped with the above tools, we give a proof of Theorem~\ref{uniformcover}.

\begin{proofof}{Theorem~\ref{uniformcover}}
Let $G=(V,E)$ be a $4$-regular $4$-edge connected multigraph and let $e=uv$ be an arbitrary edge in $G$.
We prove this theorem via induction on $n := |V(G)|$.
The base case $n = 2$ corresponds to a pair of vertices having four parallel edges, call them $e,f,g,h$.
Observe that $\frac23 \chi^{E \setminus \{e\}} = \frac13 \bigl( \chi^{\{f,g\}} + \chi^{\{f,h\}} + \chi^{\{g,h\}} \bigr)$, so the induction hypothesis is true for the base case.

For the induction step, suppose that $n \geq 3$ and the hypothesis holds for all $4$-regular $4$-edge connected multigraphs with at most $n-1$ vertices and for all choices of $e \in E$.
Consider a $4$-regular $4$-edge connected multigraph $G$ on $n$ vertices and an arbitrary edge $e = uv \in E(G)$.
Besides $e$, let $vx,vy,vz$ be the other three edges incident to $v$.
With a relabeling of vertices, by Lemma~\ref{lem:twooptions}, we may assume that $(uv,vx)$ and $(uv,vy)$ form an admissible pair in $G$ (see Figure~\ref{fig:fourcases}).

{
\begin{figure}[hbt]
    \begin{subfigure}{0.20\textwidth}
	\centering
	\begin{tikzpicture}[scale=0.70]
			\begin{scope}[every node/.style={draw=none}]
			\node (u) at (2.5,4.5) {$u$};
			\node  (v) at (2.5,2.5) {$v$};
			\node  (x) at (0,0) {$x$};
			\node  (y) at (2,0) {$y$};
			\node  (z) at (4,0) {$z$};	
			\end{scope}
			
			\begin{scope}[every	node/.style={draw,circle,minimum size=1mm,inner sep=0pt,outer sep=0pt, fill=black}]
			\node (uu) at (2,4.5) {};
			\node  (vv) at (2,2.5) {};
			\node  (xx) at (0,0.5) {};
			\node  (zz) at (2,0.5) {};
			\node  (yy) at (4,0.5) {};	
			\end{scope}
			\begin{scope}[
				every edge/.style={draw=black}]
			\path [-] (uu) edge node[left] {$e$} (vv);
			\path [-] (zz) edge (vv);
			\path [-] (xx) edge (vv);
			\path [-] (yy) edge (vv);
			\end{scope}
		\end{tikzpicture}
	\caption{\centering
	$v$ has four distinct neighbors \newline $|\adm| \in \{2,3\}.$}
	\label{case1}
    \end{subfigure}
    \hspace*{\fill}
    \begin{subfigure}{0.20\textwidth}
	\centering
	\begin{tikzpicture}[scale=0.70]
			\begin{scope}[every node/.style={draw=none}]
			\node (u) at (2.8,4.5) {$u=z$};
			\node  (v) at (2.5,2.5) {$v$};
			\node  (x) at (0,0) {$x$};
			\node  (z) at (4,0) {$y$};	
			\end{scope}
			
			\begin{scope}[every	node/.style={draw,circle,minimum size=1mm,inner sep=0pt,outer sep=0pt, fill=black}]
			\node (uu) at (2,4.5) {};
			\node  (vv) at (2,2.5) {};
			\node  (xx) at (0,0.5) {};
			\node  (yy) at (4,0.5) {};	
			\end{scope}
			
			\begin{scope}[
				every edge/.style={draw=black}]
			\path [-] (uu) edge[bend right] node[left] {$e$} (vv);
			\path [-] (uu) edge[bend left](vv);
			\path [-] (xx) edge (vv);
			\path [-] (yy) edge (vv);
			\end{scope}
		\end{tikzpicture}
	\caption{\centering $v$ has two parallel edges with $u$ \newline $\adm = \{vx,vy\}$.}
	\label{case2}
    \end{subfigure}
    \hspace*{\fill}
    \begin{subfigure}{0.20\textwidth}
	\centering
	\begin{tikzpicture}[scale=0.70]
			\begin{scope}[every node/.style={draw=none}]
			\node (u) at (2.5,4.5) {$u$};
			\node  (v) at (2.5,2.5) {$v$};
			\node  (x) at (0,0) {$x=y$};
			\node  (z) at (4,0) {$z$};	
			\end{scope}
			
			\begin{scope}[every	node/.style={draw,circle,minimum size=1mm,inner sep=0pt,outer sep=0pt, fill=black}]
			\node (uu) at (2,4.5) {};
			\node  (vv) at (2,2.5) {};
			\node  (xx) at (0,0.5) {};
			\node  (yy) at (4,0.5) {};	
			\end{scope}
			
			\begin{scope}[
				every edge/.style={draw=black}]
			\path [-] (uu) edge node[left] {$e$} (vv);
			\path [-] (xx) edge[bend left] (vv);
			\path [-] (xx) edge[bend right] (vv);
			\path [-] (yy) edge (vv);
			\end{scope}
		\end{tikzpicture}
	\caption{\centering $v$ has two parallel edges with $x, x \neq u$ \newline $\adm=\{vx,vy\}$.}
	\label{case3}
    \end{subfigure}
    \hspace*{\fill}
    \begin{subfigure}{0.20\textwidth}
	\centering
	\begin{tikzpicture}[scale=0.70]
			\begin{scope}[every node/.style={draw=none}]
			\node (u) at (2.8,4.5) {$u = z$};
			\node  (v) at (2.5,2.5) {$v$};
			\node  (x) at (2,0) {$x = y$};
			\end{scope}
			
			\begin{scope}[every	node/.style={draw,circle,minimum size=1mm,inner sep=0pt,outer sep=0pt, fill=black}]
			\node (uu) at (2,4.5) {};
			\node  (vv) at (2,2.5) {};
			\node  (xx) at (2,0.5) {};
			\end{scope}
			\begin{scope}[
				every edge/.style={draw=black}]
			\path [-] (uu) edge[bend right] node[left] {$e$} (vv);
			\path [-] (uu) edge[bend left] (vv);
			\path [-] (xx) edge[bend left] (vv);
			\path [-] (xx) edge[bend right] (vv);
			\end{scope}
		\end{tikzpicture}
	\caption{\centering $v$ has two parallel edges to each of $\{u,x\}$ \newline $\adm=\{vx,vy\}$.}
	\label{case4}
    \end{subfigure}
    \hspace*{\fill}
\caption{Four configurations of edges in $\delta(v) = \{ uv, vx, vy, vz \}$ that can arise in our proof.}
\label{fig:fourcases}
\end{figure}
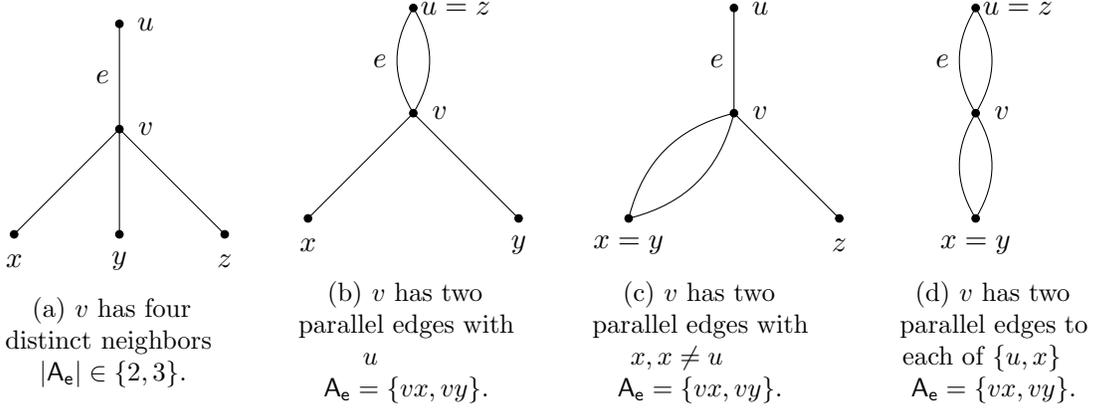
}

By the second conclusion of Lemma~\ref{lem:twooptions}, $(vy,vz)$ is an admissible pair in $G_{u,x}$, and $(vx,vz)$ is an admissible pair in $G_{u,y}$.
Consider the graph $G_1$ obtained by splitting off the pair $(vy,vz)$ in $G_{u,x}$ i.e., $G_1 = G - v + \{ux,yz\}$; it is customary to drop the vertex $v$ after all its edges have been split off.
Similarly, let $G_2$ be the graph obtained by splitting off the pair $(vx,vz)$ in $G_{u,y}$ i.e., $G_2 = G - v + \{uy,xz\}$.

Since we only split off admissible pairs, both $G_1$ and $G_2$ are $4$-regular $4$-edge connected multigraphs on $n-1$ vertices.
Recall that for any subgraph $K$ of some graph, $\chi^K$ is a shorthand for $\chi^{E(K)}$ whenever the underlying graph is clear from the context.
Applying the induction hypothesis to $G_1$ with the designated edge $e_1 = ux$ gives:
\begin{equation} \tag{ConvexComb$\mhyphen G_1$} \label{eq:comb1}
\frac23 \cdot \chi^{E(G_1) \setminus \{ e_1 \}} = \frac23 \cdot \chi^{(E \setminus \delta(v)) \cup \{yz\}} = \sum_{i=1}^{k_1} \cc^1_i \chi^{H^1_i} \ ,
\end{equation}
where $\{\cc^1_i\}_i$ denote the coefficients in a convex combination, and $\{H^1_i\}_i$ are $2$-edge connected spanning subgraphs of $G_1$ such that none of them use more than one copy of an edge in $G_1$.
Repeating the same argument for $G_2$ with the designated edge $e_2 = uy$ gives:
\begin{equation} \tag{ConvexComb$\mhyphen G_2$} \label{eq:comb2}
\frac23 \cdot \chi^{E(G_2) \setminus \{ e_2 \}} = \frac23 \cdot \chi^{(E \setminus \delta(v)) \cup \{xz\}} = \sum_{i=1}^{k_2} \cc^2_i \chi^{H^2_i} \ ,
\end{equation}
where $\{\cc^2_i\}_i$ denote the coefficients in the other convex combination arising from $\{H^2_i\}_i$.
It remains to combine \eqref{eq:comb1}~and~\eqref{eq:comb2} to obtain such a representation for $G$ with the designated edge $e$. 
We mimic the strategy from \cite{CarrR98}.

For each $i \in \{1,\dots,k_1\}$, we lift $H^1_i$ to a spanning subgraph $\hH^1_i$ of $G-e$.
Define $\hH^1_i$ as follows:
\begin{equation} \tag{Lift$\mhyphen G_1$} \label{eq:lift1}
\begin{aligned}
\hH^1_i & := \begin{cases} H^1_i - yz + vy + vz & \text{ if } yz \in E(H^1_i), \\ H^1_i + vy + vx &  \text{ if } yz \notin E(H^1_i).
\end{cases}
\end{aligned}
\end{equation}

Similarly, for each $i \in \{1,\dots,k_2\}$, we define $\hH^2_i$ as the following spanning subgraph of $G-e$:
\begin{equation} \tag{Lift$\mhyphen G_2$} \label{eq:lift2}
\begin{aligned}
\hH^2_i & := \begin{cases} H^2_i - xz + vx + vz & \text{ if } xz \in E(H^2_i), \\ H^2_i + vx + vy & \text{ if } xz \notin E(H^2_i).
\end{cases}
\end{aligned}
\end{equation}

We finish the proof of Theorem~\ref{uniformcover} by arguing that the following convex combination meets all the requirements:
\begin{equation} \tag{ConvexComb$\mhyphen G$} \label{eq:comb3}
\ccvec := \frac12 \sum_{i=1}^{k_1} \cc^1_i \chi^{\hH^1_i} + \frac12 \sum_{i=1}^{k_2} \cc^2_i \chi^{\hH^2_i} \ .
\end{equation}

Many of our arguments are the same for $G_1$ and $G_2$ so we just mention them in the context of $G_1$.
First of all, by the induction hypothesis and \eqref{eq:lift1} it is clear that $e (= uv), yz, ux \notin E(\hH^1_i)$, where $yz$ and $ux$ refer to the edges that originated from the splitting off operations applied at $v$.
Next, we argue that $\hH^1_i$ is a spanning subgraph of $G$ that uses no more than one copy of any edge in $G$.
By the induction hypothesis, none of the subgraphs $H^1_i$ use more than one copy of an edge in $G_1$, and $H^1_i$ spans $V \setminus \{v\}$. 
By the way we lift $H^1_i$ to $\hH^1_i$, it is clear that $\hH^1_i$ uses no more than one copy of any multiedge in $G$, and that it is spanning.
To see that $\hH^1_i$ is $2$-edge connected, observe that the two cases of lifting may be viewed as either (i) subdividing the edge $yz$ by a node $v$ when $yz \in E(H^1_i)$, or (ii) adding an edge $yx$ and subdividing it by a node $v$ when $yz \notin E(H^1_i)$. 
Clearly, these operations preserve $2$-edge connectivity, hence, $\hH^1_i$ is $2$-edge connected.

It remains to argue that the vector $\ccvec$ in the expression \eqref{eq:comb3} matches the vector $\frac23 \chi^{E(G) \setminus \{e\}}$.
Since $\{\cc^1_i\}_i$ and $\{\cc^2_i\}_i$ denote coefficients in a convex combination, taking an unweighted average of these two combinations gives us another convex combination.
Since none of the edges in $E(G) \setminus \delta(v)$ are modified in the lifting step, $\ccvec_f = 2/3$ for any such edge $f$. 
Next, consider the edge $vy$.
Observe that $\hH^1_i$ always contains the edge $vy$, whereas $\hH^2_i$ contains $vy$ only when $xz \notin E(H^2_i)$ (this happens with weight $1/3$).
Therefore, $\ccvec_{vy} = \frac12 \cdot 1 + \frac12 \cdot \frac13 = \frac23$.
The analysis for $vx$ is symmetric.
Lastly, consider the edge $vz$. It appears in $\hH^1_i$ ($\hH^2_i$) if and only if $yz \in E(H^1_i)$ (respectively, $xz \in E(H^2_i)$.
Therefore, $\ccvec_{vz} = \frac12 \cdot \frac23 + \frac12 \cdot \frac23 = \frac23$.
This completes the proof of Theorem~\ref{uniformcover}.
\end{proofof}

\section{Our Algorithm and the Proof of Theorem~\ref{thm:twothirds}} \label{algo}

In this section we give a proof of Theorem~\ref{thm:twothirds} which we use to obtain a $\frac43$-approximation algorithm for $\twoecm$ on half-integral instances (Theorem~\ref{thm:fourthirdsalgo}). 
We apply the same splitting-off theorem of \cite{BangJensenGJS99} together with an induction scheme that is captured in Theorem~\ref{thm:twothirds}. 
A key feature of this theorem is that we allow edges of negative cost, although the edge costs in any instance of $\twoecm$ are non-negative.

Consider a $4$-regular $4$-edge~connected multigraph $G = (V,E)$ on $n$ vertices, and let $e = uv$ be an edge in $G$.
Let $c : E \to \R$ be an arbitrary real-valued cost function.
Our goal is to obtain a $2$-edge connected spanning subgraph $H$ of $G$ whose cost is at most $\frac23 c(G-e)$ while ensuring that $H$ uses no more than one copy of any multiedge in $G$.
Observe that if we had access to the collection $\{H_1,\dots,H_k\}$ of $2$-edge connected spanning subgraphs from Theorem~\ref{uniformcover} for some $k$ that is polynomial in $|V(G)|$, then we would be done: for any cost function $c$, the cheapest subgraph in this collection (w.r.t. cost $c$) is one such desired subgraph.
It is not clear how to efficiently obtain such a collection; a naive algorithm that follows the proof of Theorem~\ref{uniformcover} does not run in polynomial time.

As alluded to before, for the purposes of obtaining a cheap $2$-edge connected subgraph, it suffices to only recurse on one of the two subinstances that arise in the proof of Theorem~\ref{uniformcover}.
This insight comes from working backwards from \eqref{eq:comb3}.
Since this convex combination for $G$ is a simple average of the convex combinations from the two subinstances (see \eqref{eq:comb1} and \eqref{eq:comb2}), it is judicious to only recurse on the ``cheaper'' subinstance.
Combining \eqref{eq:comb1} and \eqref{eq:lift1}, we get that the first subinstance gives rise to a convex combination for $\frac23 \chi^{E(G) \setminus \{e\}} + \frac13 ( \chi^{\{vy\}} - \chi^{\{vx\}} )$.
On the other hand, the second subinstance gives rise to a convex combination for $\frac23 \chi^{E(G) \setminus \{e\}} + \frac13 ( \chi^{\{vx\}} - \chi^{\{vy\}} )$.
Thus, we should recurse on $G_1$ if $c_{vx} \geq c_{vy}$, and $G_2$ otherwise.
For the sake of argument, suppose that we are recursing on $G_1$.
So far, we have ignored an important detail in the recursion: the splitting-off operation creates a new edge $yz$ that was not originally present in $G$, so we need to assign it some cost to apply the algorithm recursively. 
Depending on how we choose the cost of $yz$, it might either be included or excluded from the subgraph obtained for the smaller instance, so to bound the cost of the lifted solution we must have a handle on both outcomes of the lift operation.
Setting $c_{yz} := c_{vz} - c_{vx}$ balances the cost of both outcomes.
Note that $c_{yz}$ could possibly be negative, but this is permissible since the statement of Theorem~\ref{thm:twothirds} allows for arbitrary edge costs.
We formalize the above ideas.

In the recursive step, we pick one end vertex $v$ of $e$ and apply a complete splitting off operation at $v$ to obtain a $4$-regular $4$-edge connected graph on $n-1$ vertices; this can be implemented in $O(n)$ time.
The running time of the algorithm is $O(n^2)$, since we apply the induction step $O(n)$ times.
We remark that the running time of the algorithm in Theorem~\ref{thm:twothirds} can be improved to $O(n^{1+o(1)})$ by using the results for maintaining $3$-edge connectivity from the work of Jin and Sun \cite{JinS20}; we defer the details to a subsequent full version of our paper.

Let $T=\{v,x,y,z\}$ be the four neighbors of $v$ and let $e=uv$. Recall that $A_e$ denotes the set of edges $f\in \delta(v)\setminus \{e\}$ such that $(e,f)$ is an admissible pair (see Definition~\ref{admissibledefn}).

\begin{lemma} \label{checkadmissible}
   For $vx\in \delta(v)\setminus\{e\}$, we can check whether $vx\in \adm$ in $O(n)$ time.
\end{lemma}
\begin{proof}
   We may suppose that the elements of the set $T$ of neighbors of $v$ are all distinct. 
   Otherwise, by Lemma \ref{lem:twooptions}, we know exactly which pairs are admissible, see Figure~\ref{fig:fourcases}.
   Consider the graph $\hG=(G_{u,x})_{y,z}$ obtained by splitting off
   the pairs $(uv,vx)$ and $(yv,vz)$ at $v$.
   Let $G^*$ be the graph obtained from $\hG$
   by contracting $ux$ to a single vertex $s$ and contracting $yz$ to a
   single vertex $t$.
   Then we apply a max $s-t$ flow computation to check whether $G^*$
   has $\ge4$ edge-disjoint $s-t$ paths; otherwise, $G^*$ has an
   $(s,t)$-cut $\delta(S)$ of size $\le3$. In the latter case, it is
   clear that our trial splitting is not admissible.

   In the former case, we claim that our trial splitting is
   admissible.  Suppose that $\hG$ is not $4$-edge connected.  Then there
   exists a non-empty, proper vertex set $S$ in $\hG$ such that $|T\cap
   S| \leq |T\setminus S|$ and $|\delta_{\hG}(S)|<4$.  
   Clearly, $|S \cap T|\leq2$, 
   and if $|S \cap T|=2$, then we have $|S\cap\{u,x\}|=1$
   and $|S\cap\{y,z\}|=1$ (otherwise, $S$ would give an $(s,t)$-cut
   of $G^*$ of size $\leq3$).  Since the size of the cut of $S$ is
   the same in $G$ and in $\hG$, we have, by $4$-edge connectivity of $G$, 
   $4>|\delta_{\hG}(S)|=|\delta_{G}(S)|\ge 4$, a contradiction. 

   To see that the running time is linear, observe that $G^*$ has
   $\leq 2n$ edges, an $s-t$ flow of value $\geq4$ can be computed
   by finding $4$ augmenting paths, and each augmenting path can be
   found in linear time.
\end{proof}

\begin{proofof}{Theorem~\ref{thm:twothirds}}
First, consider the base case in the recursion when $n=2$.
The only such $4$-regular $4$-edge connected multigraph is given by four parallel edges between $u$ and $v$, of which $e$ is one.
Picking the two cheapest edges from the remaining three edges gives
the desired subgraph.

For the induction step, suppose that $n \geq 3$ and the induction hypothesis holds for all $4$-regular $4$-edge connected multigraphs with at
most $n-1$ vertices and for all choices of edge $e$.
Consider a $4$-regular $4$-edge connected multigraph $G$ on $n$ vertices and an edge $e=uv$ in $G$.

Our algorithm proceeds as follows.
By Lemmas~\ref{lem:twooptions}~and~\ref{checkadmissible}, we can find in $O(n)$-time two neighbors of $v$, say $x$ and $y$, such that  $vx, vy \in A_e$ and $c_{vx}\ge c_{vy}.$
Next, we construct the graph $\hG := (G_{u,x})_{y,z} = G - v + \{ ux,yz \}$ and extend the cost function $c$ to the new edge $yz$ as $c_{yz} := c_{vz}-c_{vx}$ (note that the cost of $ux$ is inconsequential and that $c_{yz}$ may be negative or non-negative).
We recursively find a $2$-edge connected spanning subgraph $\hH$ of $\hG$
with cost at most $\frac{2}{3}c(\hG-ux)$.
Then, we lift $\hH$ to obtain a spanning subgraph $H$ of $G$:
\begin{equation*}
\begin{aligned}
H & := \begin{cases} \hH - yz + vy + vz & \text{ if } yz \in E(\hH), \\ \hH + vy + vx &  \text{ if } yz \notin E(\hH).
\end{cases}
\end{aligned}
\end{equation*}

We analyze the cost of this subgraph.
Regardless of the cases above, our choice of $c_{yz}$ implies that $c(H)=c(\hH)+c_{vy}+c_{vx}$.
Therefore,
\begin{multline*}
c(H) \leq \frac{2}{3} c(\hG-ux)+c_{vy}+c_{vx} = \frac{2}{3}\left\{ c(G-e) - c_{vx} - c_{vy} - c_{vz} + (c_{vz} - c_{vx})\right\} + c_{vy} + c_{vx} \\
= \frac{2}{3} c(G-e) + \frac{1}{3}(c_{vy}-c_{vx}) \leq \frac{2}{3} c(G-e) \ ,
\end{multline*}
where the last inequality follows from our choice of $vx,vy$ to satisfy $c_{vx} \geq c_{vy}$.

It remains to argue that $H$ is a $2$-edge connected spanning subgraph of $G-e$ that uses no more than one copy of any multiedge in $G$.
It is clear that the following hold: (a)~$e \notin E(H)$; (b)~$H$ is a spanning subgraph of $G$; and (c)~each multiedge of $G$ appears at most once in $H$.
Since $\hH$ is $2$-edge connected and adding and/or subdividing an edge preserves $2$-edge connectivity, $H$ is $2$-edge connected.
Overall, in $O(n^2)$-time we have constructed a $2$-edge connected spanning subgraph $H$ of $G-e$ whose cost is at most $\frac23 c(G-e)$, thereby proving Theorem~\ref{thm:twothirds}.
\end{proofof}

Using Theorem~\ref{thm:twothirds}, we give a deterministic $\frac43$-approximation algorithm for $\twoecm$ on half-integral instances.

\begin{proofof}{Theorem~\ref{thm:fourthirdsalgo}}
Let $x$ be an optimal half-integral solution to $(\tsplp)$ (and $(\twoecmlp)$) for an instance given by an $n$-vertex graph $\bG = (\bV,\bE)$ and a metric cost function $c$.
Let $G = (V,E)$ denote the graph induced by $2x$ where for each $e \in E$ we include $2 x_e$ copies of the edge $e$ in $G$.
Since $x$ has (fractional) degree $2$ at each vertex and it is fractionally $2$-edge connected, $G$ is a $4$-regular $4$-edge connected multigraph.
With a slight abuse of notation, we use the same cost function for the edges of $E$: for any $e \in E$, $c_e := c_f$, where $f$ denotes the edge in $\bE$ that gave rise to $e$.
We invoke Theorem~\ref{thm:twothirds} on $G$ and some edge $e \in E$.
This gives us a $2$-edge connected spanning subgraph $H$ of $G-e$ satisfying 
${c(H) \leq \frac23 c(G-e)}$.
Lifting the subgraph $H$ to $\bG$ gives a 2-edge connected spanning multisubgraph $\bH$ (of $\bG$); note that $\bH$ uses at most two copies of any edge in $\bG$.
By the first conclusion of Theorem~\ref{thm:twothirds} and the non-negativity of $c$, $c(\bH) = c(H) \leq \frac23
c(G-e)\leq \frac23 c(G)=\frac43 c^T x$, where the last equality follows by recalling that $G$ is induced by $2x$.
Besides invoking Theorem~\ref{thm:twothirds} we only perform trivial graph operations so the running time is $O(n^2)$. 
\end{proofof}

\section{\boldmath $\twoecm$ for $3$-Regular $3$-Edge Connected Graphs} \label{sec:7over8}

Let $G=(V,E)$ be a $3$-regular $3$-edge connected graph with non-negative edge costs $c \in \Rp^E$. 
In this section we consider an analogous problem to that of Theorem~\ref{thm:twothirds}, namely the problem of finding a polynomial-time algorithm which gives a $2$-edge connected spanning multisubgraph of $G$ of cost at most $\beta c(G)$ for some $\beta \geq 0$. 
Note that the everywhere $\frac23$ vector for $G$ is feasible for $(\tsplp)$. 
For any costs $c$ for which the everywhere $\frac23$ vector is also optimal for $(\tsplp)$ (such as for the graphic metric), such an algorithm would provide a $\frac{3\beta}{2}$-approximation for $\twoecm$. 
The conjecture that $\twoecmig = \frac43$ would then imply $\beta = \frac43 \cdot \frac23 = \frac89$ should be possible, and the $\frac65$ conjecture for $\twoecmig$ would imply $\beta = \frac65 \cdot \frac23 = \frac45$ should be possible. 
In \cite{BoydL17} a constructive algorithm for $\beta = \frac45$ is given, however it does not run in polynomial time.

In \cite[Theorem~2]{HaddadanNR19}, Haddadan, Newman, and Ravi show that it is possible to do better than $\frac89$ for this problem, and provide an efficient algorithm for $\beta = \frac{15}{17}$. 
In fact, they show that the everywhere $\frac{15}{17}$ vector can be expressed as a convex combination of $2$-edge connected spanning multisubgraphs of $G$ and this convex combination can be found in polynomial time.
They remark that combining their ideas with an efficient algorithm for Theorem~\ref{uniformcover} would imply the result for $\beta = \frac{7}{8} (< \frac{15}{17})$. 
Although a polynomial-time algorithm for Theorem~\ref{uniformcover} is not currently known, it is possible to use our result in Theorem~\ref{thm:twothirds} to obtain $\beta = \frac{7}{8}$, as follows.
	
\begin{theorem} \label{thm:7over8}	
Let $G=(V,E)$ be a $3$-regular $3$-edge connected graph on $n$ vertices with non-negative edge costs $c \in \Rp^E$. 
Then in $O(n^3)$-time we can find a $2$-edge connected spanning multisubgraph $H$ of $G$ such that $c(H) \leq \frac78 c(G)$.
\end{theorem}

\begin{proof}
Let $F$ be a $2$-factor of $G$ that intersects all of the $3$-edge cuts and $4$-edge cuts of $G$. 
Such a $2$-factor can be found in $O(n^3)$-time (see \cite[Theorem~5.4]{BoydIT13}).
Let $G'$ be the graph obtained by contracting the cycles of $F$ and removing any resulting loops, and let $M : = E(G')$. 
Clearly $G'$ is $5$-edge connected (by choice of $F$), and thus the vector $y \in \Rp^M$ defined by $y_e := \frac25$ for all $e \in M$ is feasible for $(\twoecmlp)$ for $G'$. 
It then follows from the polyhedral analysis of Wolsey \cite{Wolsey80} and Shmoys and Williamson \cite{ShmoysW90} of the $(\tsplp)$ that we can find a $2$-edge connected spanning multisubgraph of $G'$ with edge set $R$ satisfying $c(R) \leq \frac32 c^T y = \frac35 c(M)$.
Then the graph $H_1$ induced by $F \cup R$ is a $2$-edge connected spanning multisubgraph of $G$ such that 

\begin{equation} \label{eq:H1}
c(H_1) \leq c(F) + \frac35 c(M) \leq c(F) + \frac35 c(E \setminus F) \ .
\end{equation}

Now consider the vector $z \in \Rp^E$ where $z_e := 1/2$ for all $e \in F$, and $z_e := 1$ otherwise. 
Vector $z$ is a feasible half-integer solution for $(\tsplp)$, and thus by Theorem~\ref{thm:twothirds} and the ideas used in the proof of Theorem~\ref{thm:fourthirdsalgo}, in $O(n^2)$-time we can find a $2$-edge connected spanning multisubgraph $H_2$ of $G$ such that 

\begin{equation} \label{eq:H2}
c(H_2) \leq \frac23 c(F) + \frac43 c(E \setminus F).
\end{equation}

We complete the proof by showing that either $H_1$ or $H_2$ has cost at most $\frac78 c(G)$. 
Using (\ref{eq:H1}) and (\ref{eq:H2}) we have:

\begin{equation*}
\min(c(H_1),c(H_2)) \leq \frac58 c(H_1) + \frac38 c(H_2) \leq \frac58 \bigl( c(F)+\frac35 c(E \setminus F) \bigr) + \frac38 \bigl( \frac23 c(F) + \frac43 c(E \setminus F) \bigr) = \frac78 c(G).
\end{equation*}
\end{proof}

\section*{Acknowledgments}
We thank Chaitanya Swamy for pointing us to \cite{CarrV02}.

\bibliography{four_thirds_2ecm_ref}

\end{document}